\documentclass[letterpaper, 10 pt, conference]{ieeeconf}  % Comment this line out
\IEEEoverridecommandlockouts                              % This command is only
\usepackage{times} % assumes new font selection scheme installed
\usepackage{amsmath} % assumes amsmath package installed
\usepackage{amssymb}  % assumes amsmath package installed

\let\labelindent\relax % To solve error with enumitem...

\usepackage[hidelinks]{hyperref}
\usepackage[noabbrev]{cleveref}
\usepackage{enumitem}
\usepackage{graphicx}
\usepackage{multirow}
\usepackage{physics}
\usepackage{subcaption}
\usepackage{xcolor}

\newcommand{\gren}{$\gamma$REN}

\newtheorem{definition}{Definition}%[section]
\newtheorem{theorem}{Theorem}%[section]
\newtheorem{remark}{Remark}%[section]
\newtheorem{example}{Example}

\newtheorem{lemma}[theorem]{Lemma}%[section]

\title{\LARGE \bf
Learning Over Contracting and Lipschitz Closed-Loops for Partially-Observed Nonlinear Systems (Extended Version)
}

\author{Nicholas H. Barbara, Ruigang Wang, and Ian R. Manchester% <-this % stops a space
\thanks{*This work was supported in part by the Australian Research Council, and the NSW Defence Innovation Network.}% <-this % stops a space
\thanks{The authors are with the Australian Centre for Robotics (ACFR), and the School of Aerospace, Mechanical and Mechatronic Engineering, The University of Sydney, Sydney, NSW 2006, Australia
        {\tt\small nicholas.barbara@sydney.edu.au}}%
}

\begin{document}

% --------------------------------------------------------------------
%
% Title and setup
%
% --------------------------------------------------------------------
\maketitle
\thispagestyle{empty}
\pagestyle{empty}

% --------------------------------------------------------------------
%
% Abstract
%
% --------------------------------------------------------------------
\begin{abstract}
 This paper presents a policy parameterization for learning-based control on nonlinear, partially-observed dynamical systems.  The parameterization is based on a nonlinear version of the Youla parameterization and the recently proposed \textit{Recurrent Equilibrium Network} (REN) class of models. We prove that the resulting Youla-REN parameterization automatically satisfies stability (contraction) and user-tunable robustness (Lipschitz) conditions on the closed-loop system. This means it can be used for safe learning-based control with no additional constraints or projections required to enforce stability or robustness. We test the new policy class in simulation on two reinforcement learning tasks: 1) magnetic suspension, and 2) inverting a rotary-arm pendulum. We find that the Youla-REN performs similarly to existing learning-based and optimal control methods while also ensuring stability and exhibiting improved robustness to adversarial disturbances.
\end{abstract}

% --------------------------------------------------------------------
%
% Introduction
%
% --------------------------------------------------------------------
\section{Introduction} \label{sec:intro}

Deep reinforcement learning (RL) has been a driving force behind many recent successes in learning-based control, with applications ranging from discrete game-like problems \cite{Mnih++2015,Silver++2017} to robotic locomotion \cite{Siekmann++2021a}. As its popularity continues to grow, there is increasing need for a learning framework that offers the stability and robustness guarantees of classical control methods while still being fast and flexible for learning in complex environments \cite{Hu++2023}. 

One promising idea is to directly learn over a set of robustly stabilizing controllers for a given dynamical system. RL policies are then guaranteed to naturally satisfy robustness and stability requirements even during training. Parameterizing the space of all such controllers is well-studied for linear systems \cite{Anderson1998}. In fact, learning over this space results in policies that perform better and are more robust than those learned from typical RL frameworks which do not consider stability \cite{Roberts++2011}. While there has been considerable work extending this parameterization to nonlinear systems with full state knowledge \cite{Imura+Yoshikawa1997,Furieri++2023} or specific structures \cite{Moore+Irlicht1992,Blanchini++2010}, the problem of general \textit{partially-observed} nonlinear systems (full state information unavailable) is more challenging. 
In this paper, we present a parameterization of robust stabilizing controllers for partially-observed nonlinear systems that can be readily applied to learning-based control.

\subsection{Previous work on linear systems} \label{sec:intro-linear}

We recently proposed the Youla-REN policy class for learning over all stabilizing controllers for partially-observed \textit{linear} systems \cite{Wang+Manchester2022,Wang++2022}. It combines the classical Youla-Kucera parameterization with the \textit{Recurrent Equilibrium Network} (REN) model architecture \cite{Revay++2023}. The Youla parameterization is an established idea in linear control theory that represents all stabilizing linear controllers for a given linear system \cite{Youla++1976}. One common construction augments an existing stabilizing ``base'' controller with a stable linear parameter $\mathcal{Q}$, which is optimized to improve the closed-loop system performance \cite{Li++2021,Mahtout2020}. We extended this idea in \cite{Wang++2022}, showing that if $\mathcal{Q}$ is a contracting and Lipschitz \textit{nonlinear} system, the Youla parameterization represents all stabilizing \textit{nonlinear} controllers for a given linear system.

A key feature of our work in \cite{Wang++2022} was using RENs for the Youla parameter $\mathcal{Q}$. The \textit{direct parameterization} presented in \cite{Revay++2023} allowed us to construct RENs that universally approximate all contracting and Lipschitz systems. This meant that we could use unconstrained optimization to train the Youla-REN. Training RENs in this way is less conservative than weight-rescaling methods such as \cite{Miller+Hardt2018}, and significantly faster than solving large semi-definite programs during training \cite{Pauli++2022} or projected gradient descent methods such as \cite{Knight+Anderson2011, Junnarkar++2023}. Our aim is to retain this computational efficiency when extending our framework to learning stabilizing controllers for nonlinear systems.
    
\subsection{Youla parameterization for nonlinear systems} \label{sec:intro-nl-youla}

Significant theoretical advances were made in the 1980s-90s to extend the Youla parameterization to partially-observed nonlinear systems. Early work by \cite{Paice+Moore1990,Chen+deFigueiredo1990} addressed the problem by using left coprime factorizations. However, state-space models of left coprime factors are limited to  nonlinear systems with specific structures \cite{Moore+Irlicht1992}. A more extensive framework was proposed by \cite{Paice+vanderSchaft1994,Paice+vanderSchaft1996,Fujimoto+Sugie2000,Fujimoto+Sugie1998} using kernel representations to parameterize all stabilizing controllers for nonlinear systems based on input-to-state stability.
Despite providing a rather general result, kernel representations are often not intuitive to work with in practical learning-based control.
Moreover, the focus of these works was on stabilizing a particular equilibrium state. Many applications require control systems that track a wide range of reference trajectories. This motivates the need for a framework with a sufficiently strong and flexible notion of stability that is also intuitive to implement in practice.

\subsection{Contributions}
In this paper, we extend the Youla-REN policy class proposed in \cite{Wang+Manchester2022,Wang++2022} to nonlinear systems and address the open questions outlined in Section~\ref{sec:intro}. In particular, we:
\begin{enumerate}
    \item Parameterize robust (Lipschitz) stabilizing feedback controllers for \textit{partially-observed} nonlinear systems based on a strong stability notion --- contraction.
    \item Demonstrate that the Youla-REN can be applied to learning-based control by investigating its performance on  two simulated RL tasks: 1) magnetic suspension and 2) inverting a rotary-arm pendulum.
    \item Show through simulation that we can intuitively tune the trade-off between performance and robustness of the learned policy via the direct parameterization of RENs.
\end{enumerate}

% Notation
\subsection{Notation}
Consider the set of sequences $\ell_{2e}^n = \{x \mid x: \mathbb{N} \rightarrow \mathbb{R}^n\}$, where $n$ is omitted if it can be inferred from the context. For any $x \in \ell_{2e}^n$, write $x_t \in \mathbb{R}^n$ for the value of the sequence at time $t \in \mathbb{N}$. We define the subset $\ell_2 \subset \ell_{2e}^n$ as the set of all square-summable sequences such that $x\in \ell_2 \Leftrightarrow \|x\| := \sqrt{\sum_{t=0}^{\infty} |x_t|^2}$ is finite, where $|\cdot|$ denotes the Euclidean norm. We also define the norm of the truncation of $x\in\ell_2$ over $[0,T]$ as $\|x\|_T := \sqrt{\sum_{t=0}^{T} |x_t|^2}$ for all $T\in\mathbb{N}$ such that $x\in \ell_{2e}^n \iff \|x\|_T$ is finite. A function $f:\mathbb{R}^n\rightarrow \mathbb{R}^m$ is Lipschitz continuous if there exist a constant $\gamma\in\mathbb{R}^+$ such that $|f(a)-f(b)|\leq \gamma |a-b|, \forall a, b\in \mathbb{R}^n$.

% Definitions
\subsection{Definitions}
They main results in this paper concern the analysis of contracting (``stable'') and Lipschitz (``robust'') dynamical systems. Consider a system
\begin{equation} \label{eqn:generic-system}
    \mathcal{T}: 
    \begin{cases}
        x_{t+1} = f(x_t, u_t) \\
        y_t = h(x_t, u_t)
    \end{cases}
    \end{equation}
with state $x_t\in\mathbb{R}^n$, inputs $u_t\in\mathbb{R}^m$, and outputs $y_t \in \mathbb{R}^p$.

\begin{definition}[Contraction] \label{dfn:contraction}
$\mathcal{T}$ is contracting if there exists a smooth (continuously differentiable) function $V:\mathbb{R}^n \times \mathbb{R}^n \rightarrow \mathbb{R}^+$ which, for any fixed input sequence $u \in \ell_{2e}^m$, satisfies
\begin{align}
    c_1 \abs{x^1_t - x^2_t}^2 \le V(x^1_t, &x^2_t) \le c_2 \abs{x^1_t - x^2_t}^2 \label{eqn:contract-1}\\
    V(x^1_{t+1}, x^2_{t+1}) &\le \alpha V(x^1_t, x^2_t)\label{eqn:contract-2}
\end{align}
where $c_2 \ge c_1 > 0$ and $\alpha \in [0,1)$ is the contraction rate.
\end{definition}

Intuitively, a contracting system is one that exponentially forgets its initial conditions. We introduce a slightly weaker notion of contraction based only on exponential convergence. 

\begin{definition}[Contraction with transients] \label{dfn:contraction-transients}
    $\mathcal{T}$ is contracting with transients at rate $\alpha \in [0,1)$ if for any two initial conditions $x_0^1, x_0^2 \in \mathbb{R}^n$ and for any fixed input sequence $u \in \ell_{2e}^m$, there exists $\beta:\mathbb{R}^n\times \mathbb{R}^n\rightarrow \mathbb{R}^+$ such that
    \begin{equation} \label{eqn:contraction}
        |x^1_t - x^2_t| \le \beta(x_0^1,x_0^2) \cdot \alpha^t \qquad \forall \, t \in \mathbb{N}.
    \end{equation}
    Note that the overshoot $\beta$ is a function of initial conditions.
\end{definition}

\begin{definition}[Lipschitz with transients] \label{dfn:lipschitz-sys}
    A system $\mathcal{T}$ is Lipschitz with transients (referred to simply as \textit{Lipschitz} in this paper) if for any two input sequences $u^1, u^2 \in \ell_{2e}^m$ and initial conditions $x^1_0, x^2_0 \in \mathbb{R}^n$ we have 
    \begin{equation} \label{eqn:lipschitz-sys}
        \norm{y^1 - y^2}_T \le \gamma \norm{u^1 - u^2}_T + \kappa(x^1_0, x^2_0) \ \forall \, T \in \mathbb{N}
    \end{equation}
    where $\gamma \in \mathbb{R}^+$ is the Lipschitz constant and $\kappa(x^1_0, x^2_0) \ge 0$.
\end{definition}

A system with a smaller Lipschitz bound is more robust to sudden changes in its inputs.
Adding the condition $\kappa(a,a) = 0$ for any $a \in \mathbb{R}^n$ would recover the definition for a system with an incremental $\ell^2$ gain bound of $\gamma$ from \cite{Revay++2023,Revay++2021a}.

% --------------------------------------------------------------------
%
% Problem statement
%
% --------------------------------------------------------------------

\section{Problem statement}

Consider a nonlinear dynamical system $\mathcal{G}$ described by
\begin{equation} \label{eqn:system-G}
    \mathcal{G}: 
    \begin{cases}
        x_{t+1} = f(x_t, u_t + r_t) \\
        y_t = c(x_t)
    \end{cases}
\end{equation}
with internal states $x_t \in \mathbb{R}^n$, controlled inputs $u_t \in \mathbb{R}^m$, and measured outputs $y_t \in \mathbb{R}^p$. The control signal is perturbed by a known, bounded input $r_t \in \mathbb{R}^m$ (e.g., a reference signal) satisfying $\abs{r_t} \le \overline{r}$ for all $t\in\mathbb{N}$ with $\overline{r}\in\mathbb{R}^+$. We denote the total inputs as $ \bar{u}_t= u_t  + r_t$. The exogenous inputs and controlled outputs are $r_t$ and $z_t = (x_t^\top, u_t^\top)^\top$, respectively. 

Our aim is to learn feedback controllers of the form $u = \mathcal{K}_\theta (y)$ where $\theta$ is a learnable parameter. Controllers may be nonlinear and dynamical systems themselves. The closed-loop system of $\mathcal{G}$ and $\mathcal{K}_\theta$ should satisfy the following stability, robustness, and performance criteria (respectively):
\begin{enumerate}
    \item The closed-loop system is contracting (with transients) such that initial conditions are forgotten exponentially.
    \item The closed-loop response to external inputs (the map $r \mapsto z$) is Lipschitz.
    \item The controller $\mathcal{K}_\theta$ at least locally and approximately minimizes a cost function of the form
    \begin{equation} \label{eqn:general-cost-func}
        J_\theta = E \left[ \sum_{t=0}^{T-1} g(x_t, u_t) + g_T(x_T) \right]
    \end{equation}
    where the expectation is over $x_0$ and external inputs.
\end{enumerate}

% --------------------------------------------------------------------
%
% Proof: Nonlinear Youla Parameterization
%
% --------------------------------------------------------------------
\section{A nonlinear Youla parameterization} \label{sec:theory}

% --------------------------------------------------------------------
% The Youla-REN
% --------------------------------------------------------------------
\subsection{The Youla architecture}

Suppose $\mathcal{G}$ is in feedback with a \textit{base controller} $\mathcal{K}_b$ consisting of an observer and state-feedback controller:
\begin{equation} \label{eqn:base-ctrl}
    \mathcal{K}_b:
    \begin{cases}
        \hat{x}_{t+1} = f_o(\hat{x}_t, u_t+r_t, y_t) \\
        \bar{u}_t = k(\hat{x}_t)
    \end{cases}
\end{equation}
where $\hat{x}_t \in \mathbb{R}^n$ is the estimated value of the true state $x_t$, and the observer system $\mathcal{O}(\bar{u}, y)$ with state vector and output $\hat{x}$ is define by the first equation in \eqref{eqn:base-ctrl}. Denote the predicted measurements as $\hat{y}_t = c(\hat{x}_t)$ where $\hat{y} \in \mathbb{R}^p$. Our proposed controller parameterization (Fig.~\ref{fig:youla-sys}) augments the base controller with a (possibly nonlinear) system $\mathcal{Q}: (r, \tilde{y}) \mapsto \tilde{u}$, where $\tilde{y}_t = y_t - \hat{y}_t$ are the innovations. Specifically, the augmented controller is
\begin{equation} \label{eqn:youla-ctrl}
    \mathcal{K}_\mathcal{Q}: 
    \begin{cases}
        \hat{x}_{t+1} = f_o(\hat{x}_t, u_t + r_t, y_t) \\
        \bar{u}_t = k(\hat{x}_t) +\tilde{u}_t
    \end{cases}
\end{equation}
with 
\begin{equation} \label{eqn:youla-param}
    \mathcal{Q}: 
    \begin{cases}
        q_{t+1} = f_q(q_t, r_t, \tilde{y}_t) \\
        \tilde{u}_t = h_q(q_t, r_t, \tilde{y}_t),
    \end{cases}
    q_t\in\mathbb{R}^q.
\end{equation}
Here $\mathcal{K}_{\mathcal{Q}}$ is a nonlinear version of the Youla-Kucera parameterization, where $\mathcal{Q}$ is the Youla parameter.
\begin{figure}[t]
    \centering
    \includegraphics[width=0.4\textwidth]{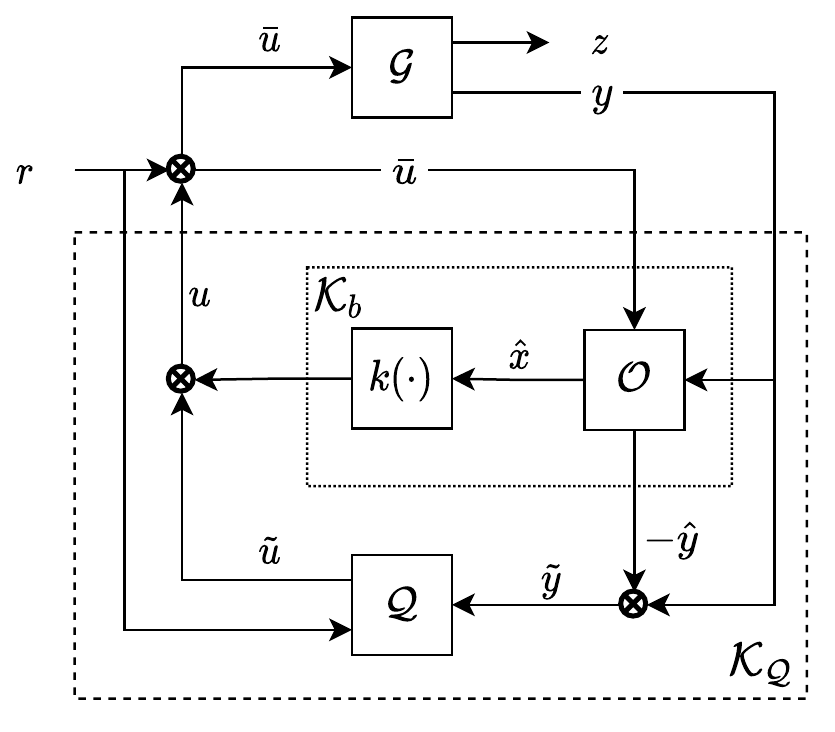}
    \caption{A version of the Youla-Kucera parameterization, where the Youla parameter $\mathcal{Q}$ augments an observer-based feedback controller $\mathcal{K}_b$.}
    \label{fig:youla-sys}
\end{figure}

% --------------------------------------------------------------------
% Assumptions
% --------------------------------------------------------------------
\subsection{Assumptions}

We make the following assumptions, drawing inspiration from \cite{Fujimoto+Sugie2000, Yi++2022}:
\begin{enumerate}[label=\textbf{A\arabic*)}]
    \item \textit{Robustly stabilizing base controller:} The closed-loop system composed of $\mathcal{G}$ in feedback with $\mathcal{K}_b$ is contracting and the map $r \mapsto z$ is Lipschitz.
    \label{A1}
    \item \textit{Observer correctness:} Given $\hat{x}_0 = x_0$, the observer exactly replicates the plant dynamics. That is, $f(x_t, \bar{u}_t) = f_o(x_t, \bar{u}_t, c(x_t)) \ \forall t\in\mathbb{N}$.
    \label{A2}
    \item \textit{Contracting \& Lipschitz observer:} The observer $\mathcal{O}$ is contracting and the map $(\bar{u},y)\mapsto \hat{x}$ is Lipschitz.
    \label{A3}
    \item \textit{Smooth maps:} All functions are Lipschitz continuous and $f$ in \eqref{eqn:system-G} is continuously differentiable in $x$ on $\mathbb{R}^n$. 
    \label{A4}
\end{enumerate}

% --------------------------------------------------------------------
% Theoretical results
% --------------------------------------------------------------------
\subsection{Theoretical results}

Our first main result is that augmenting a robustly stabilizing controller with a contracting and Lipschitz $\mathcal{Q}$ ensures the closed-loop system will also be contracting and Lipschitz. This allows optimization of the closed-loop response via $\mathcal{Q}$ while maintaining stability and robustness guarantees. 
\begin{theorem} \label{thm:forward}
    Suppose that assumptions \ref{A1} to \ref{A4} hold and the Youla parameter $\mathcal{Q}$ is contracting and Lipschitz. Then the closed-loop system of $\mathcal{G}$ in \eqref{eqn:system-G} and $\mathcal{K}_\mathcal{Q}$ in \eqref{eqn:base-ctrl} to \eqref{eqn:youla-param} is contracting \textit{with transients} and the map $r \mapsto z$ is Lipschitz.
\end{theorem} 

The proof is provided in the appendix and can be summarized as follows.
The observer error $\tilde{x}_t := x_t - \hat{x}_t$ exponentially converges to zero since the plant trajectory $x$ is a particular solution of the observer, which is a contracting and Lipschitz system. This occurs regardless of $r_t$ and $\tilde{u}_t$. To prove contraction with transients, we show that the states of a contracting system exponentially converge given exponentially converging inputs, and apply this to $\mathcal{Q}$ and the closed-loop system under $\mathcal{K}_b$ (``base system''). We repeatedly apply Lipschitz properties of $\mathcal{Q}$ and the base system  to prove that the closed-loop system under $\mathcal{K}_\mathcal{Q}$ is also Lipschitz.

One interesting question is whether the converse holds: that is, can a contracting and Lipschitz closed loop always be parameterized by a contracting and Lipschitz $\mathcal{Q}$? Our second main result shows that this is true under additional assumptions. We consider a perturbed nonlinear system 
\begin{equation} \label{eqn:plant-disturbed}
    \mathcal{G}_d: 
    \begin{cases}
        x_{t+1} = f(x_t, u_t + r_t) + d_{x_t}\\
        y_t = c(x_t) + d_{y_t}
    \end{cases}
\end{equation}
where $d_{x_t}\in \mathbb{R}^n$ and $d_{y_t}\in\mathbb{R}^p$ are additive process and measurement disturbances, respectively. The following is a stronger version of assumption \ref{A1}.
\begin{enumerate}[label=\textbf{A\arabic*)}]\addtocounter{enumi}{4}
    \item \textit{Robustness to disturbances:} The closed-loop system $(r, d_x, d_y) \mapsto z$ composed of $\mathcal{G}_d$ in feedback with $\mathcal{K}_b$ is contracting and Lipschitz.
    \label{A5}
\end{enumerate}

\begin{theorem} \label{thm:reverse}
    Suppose that assumptions \ref{A2} to \ref{A5} hold. Then, any controller $\mathcal{K}$ forming a contracting and Lipschitz closed-loop map $(r, d_x, d_y) \mapsto z$ with $\mathcal{G}_d$ in \eqref{eqn:plant-disturbed} can be represented by \eqref{eqn:base-ctrl} to \eqref{eqn:youla-param} with contracting and Lipschitz $\mathcal{Q}$.
\end{theorem}

The proof follows by augmenting the robustly stabilizing controller $\mathcal{K}$ with an observer $\mathcal{O}$ to form a map $\mathcal{Q}_\mathcal{K}: (r, \tilde{y}) \mapsto \tilde{u}$, which is contracting and Lipschitz by comparison with the closed-loop system under $\mathcal{K}$ (see appendix).

\begin{remark}
    It is worth to noting that the closed-loop system is not guaranteed to be contracting and Lipschitz under bounded but unknown additive disturbances (see Example~\ref{eg:counterexample}). This is because the observer error does not converge to zero in the presence of unknown disturbances, which is the key property required in Theorem~\ref{thm:forward}. However, the closed-loop states will still be bounded under bounded additive disturbances \cite{Lohmiller+Slotine1998,Tsukamoto++2021}.
\end{remark}
\begin{example} \label{eg:counterexample}
    Consider the following scalar system
    \[
    \begin{split}
        \mathcal{G}:& \quad x_{t+1}= 0.5 \sin(x_t) + \tilde{u}_t + d_t\\
        \mathcal{K}_{\mathcal{Q}}:& \;
        \begin{cases}
            \hat{x}_{t+1} = 0.5 \sin(\hat{x}_t) + \tilde{u}_t \\
            \tilde u_t= 10\tilde y=10(x-\hat{x}_t)
        \end{cases} 
    \end{split}
    \]
    which is contracting for the case $d_t\equiv 0$. When $d_t\equiv 1$, the system converges to multiple solutions, hence the closed-loop system is not contracting under non-zero disturbances. 
\end{example}

% --------------------------------------------------------------------
%
% Numerical Experiments
%
% --------------------------------------------------------------------
\section{Numerical Experiments} \label{sec:numeric-experiments}

We now examine the performance of the Youla-REN policy class on two RL problems: 1) magnetic suspension, and 2) inverting a rotary-arm pendulum. Each system is nonlinear and partially-observed, with different base controller designs to test the policy class under different architectures. We compare performance and robustness of three policy types:
\begin{enumerate}
    \item \textit{Youla-REN}: uses a contracting REN for the Youla parameter $\mathcal{Q}$ (see \eqref{eqn:youla-ctrl} and \eqref{eqn:youla-param}).
    \item \textit{Youla-{\gren}}: uses a REN with a Lipschitz upper bound of $\gamma$ (where $\gamma \to \infty$ recovers the contracting REN).
    \item \textit{Feedback-LSTM}: an LSTM network \cite{Hochreiter+Schmidhuber1997} augmenting the base controller via direct feedback of the measurement output, with $\tilde{u} = \mathcal{F}(y)$ for an LSTM system $\mathcal{F}$.
\end{enumerate}
Since RENs are universal approximators of contracting and Lipschitz systems \cite[Prop.~2]{Wang++2022}, then by Theorem~\ref{thm:forward} the Youla-($\gamma$)REN parameterizes a set of contracting and Lipschitz closed-loops for partially-observed nonlinear systems.
The Feedback-LSTM form is commonly used in deep RL (e.g: \cite{Siekmann++2021a}) but provides no such stability or robustness guarantees. 
For more detail on contracting and Lipschitz parameterizations of RENs, see \cite{Revay++2023}. Our experiments were written in 
Julia using \texttt{RobustEquilibriumNetworks.jl} \cite{Barbara++2023b} and are available on github\footnote{
\url{https://github.com/nic-barbara/CDC2023-YoulaREN}
}.

\subsection{Problem setup}

% Learning objective
\subsubsection{Learning objective}
Let $z_t = p(x_t)$ be a performance variable to be tracked for some function $p$. We formulated the RL tasks as minimizing a quadratic cost on the differences $\Delta z_t = z_t - z_\mathrm{ref}$, $\Delta u_t = u_t - u_\mathrm{ref}$ between performance variables and controls, and their desired reference values (respectively). That is, 
\begin{equation} \label{eqn:cost-func}
    \text{min.} \  E[J] \ \text{s.t.} \ J(x_0) = \sum_{t=0}^{T-1} \left( |\Delta z_t|^2_{Q} + |\Delta u_t|^2_{R} \right)
\end{equation}
where the cost function $J$ is weighted by matrices $Q$ and $R$. The expectation is over all possible initial conditions and random disturbances. We used $T = 100$ time samples.

\begin{figure}[!t]
    \centering
    \begin{subfigure}[b]{0.25\textwidth}
        \centering
        \includegraphics[width=\textwidth]{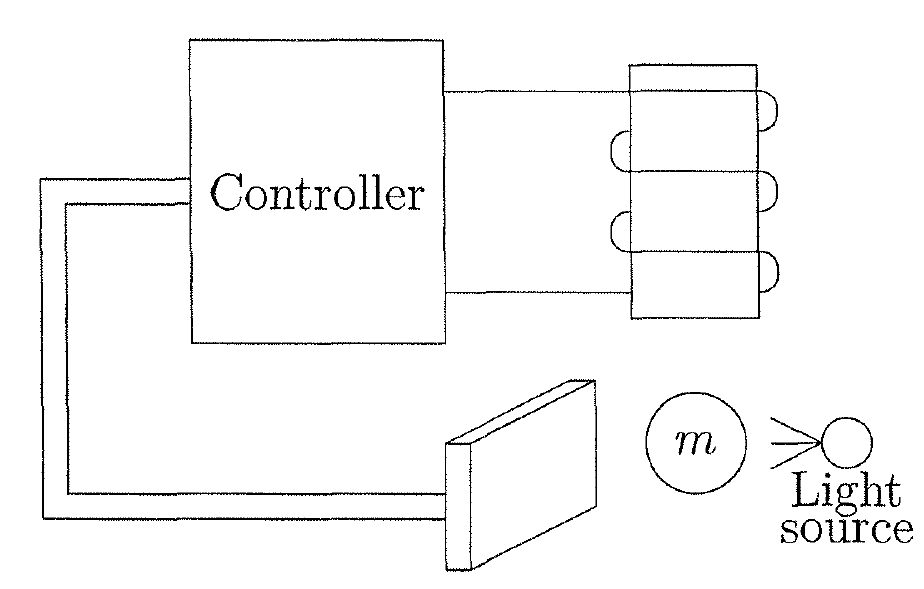}
        \caption{1D magnetic suspension.}
        \label{fig:mag-sketch}
    \end{subfigure}
    \begin{subfigure}[b]{0.2\textwidth}
        \centering
        \includegraphics[width=\textwidth]{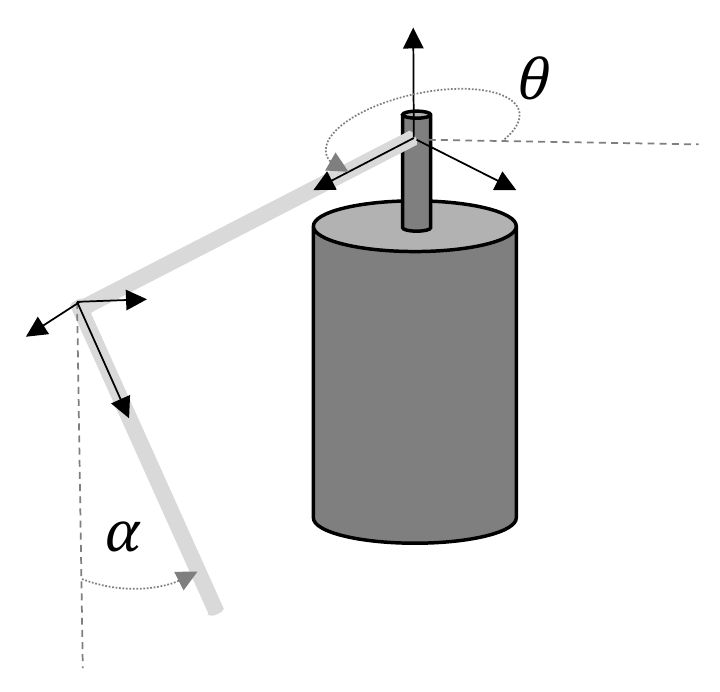}
        \caption{Rotary-arm pendulum.}
        \label{fig:qube-sketch}
    \end{subfigure}
    \caption{The two RL environments examined in Sec.~\ref{sec:numeric-experiments}: (a) a magnetic ball moving vertically under the influence of an electromagnet, and (b) a rotary-arm pendulum driven by a single motor on the rotating arm. The light source in (a) measures the ball position. Fig.~\ref{fig:mag-sketch} from \cite{Khalil2002}.}
    \label{fig:problem-setup}
\end{figure}

% --------------------------------------------------------------------
% Magnetic suspension
% --------------------------------------------------------------------
\subsubsection{Magnetic suspension}

Consider the one-dimensional magnetic suspension system presented in \cite{Khalil2002}, illustrated in Fig.~\ref{fig:mag-sketch}. The system has three states (ball position, velocity, and coil current) and one input (coil voltage). Only the ball position and coil current are measured. We used the same nonlinear system model as in Exercise 13.27 of \cite{Khalil2002}. We added random noise to all states and measurements with standard deviations $5\times10^{-4}$ and $10^{-3}$ (respectively).

The objective was to stabilize the ball at a height of 5\,cm with minimal control effort. Our base controller consisted of a high-gain observer (\cite[Sec. 14.5.2]{Khalil2002}) and a state-feedback controller designed with the backstepping and Lyapunov re-design methods outlined in \cite[Sec.~14.2-14.3]{Khalil2002}. 
We encoded the learning objective in \eqref{eqn:cost-func} with $z_t = x_t$, $Q = \mathrm{diag}(1/0.025^2, 0, 0)$, $R = 1/50^2$.

% --------------------------------------------------------------------
% Rotary-arm pendulum
% --------------------------------------------------------------------
\subsubsection{Rotary-arm pendulum}

Next we considered the rotary-arm pendulum system in Fig.~\ref{fig:qube-sketch}. The system has four states (rod angles and angular velocities) and one control input (motor voltage). Only the angles are measured. The system dynamics are presented in \cite[Eqn.~31]{Cazzolato+Prime2011}. We added noise with standard deviation $10^{-2}$ to all states and measurements.

The control objective was to stabilize the pendulum in its (unstable) upright equilibrium, again with minimal control effort. We designed a state-feedback policy consisting of an energy-pumping controller to swing the pendulum arm upwards (\cite[Eqn.~8]{Astrom+Furuta2000}) and a static linear quadratic regulator to balance the pendulum within $30^\circ$ of the vertical. We completed the base controller with a high-gain observer. The learning objective was defined as per \eqref{eqn:cost-func} with $z_t = (\cos\theta_t, \sin\theta_t, \cos\alpha_t, \sin\alpha_t)^\top$ for the arm and pendulum angles $\theta_t, \alpha_t$ (respectively), and $Q = \mathrm{diag}(5, 5, 10, 10)$, $R = 0.01$. For small deviations from vertical, this is approximately a quadratic cost on $\Delta \theta_t, \Delta \alpha_t$.

% --------------------------------------------------------------------
% Results
% --------------------------------------------------------------------

\subsection{Results and discussion}

% Loss curves
\begin{figure*}[!t]
    \centering
    \begin{subfigure}[b]{0.47\textwidth}
        \centering
        \includegraphics[width=\textwidth]{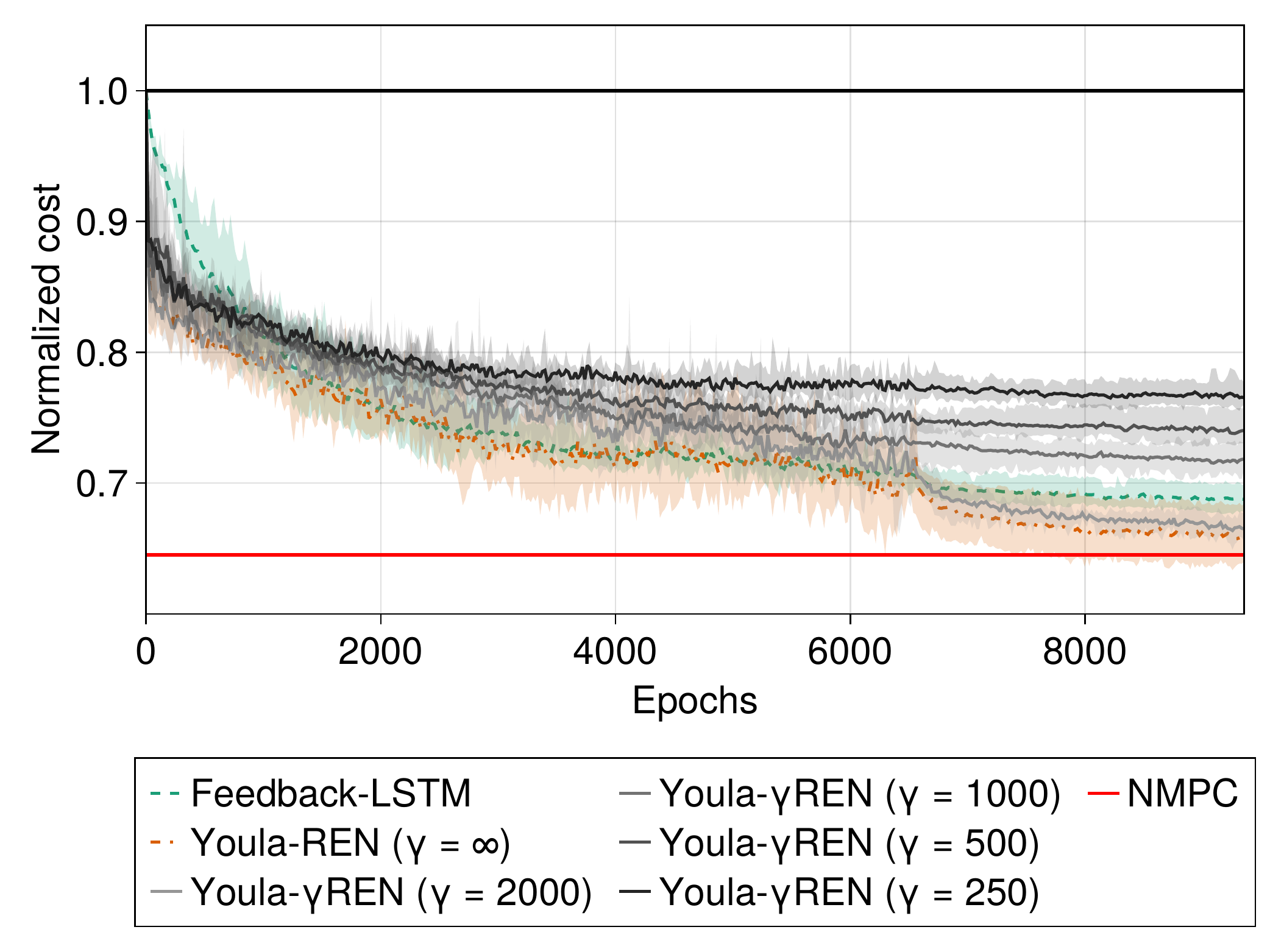}
        \caption{Magnetic suspension.}
        \label{fig:mag-learning}
    \end{subfigure}
    \begin{subfigure}[b]{0.47\textwidth}
        \centering
        \includegraphics[width=\textwidth]{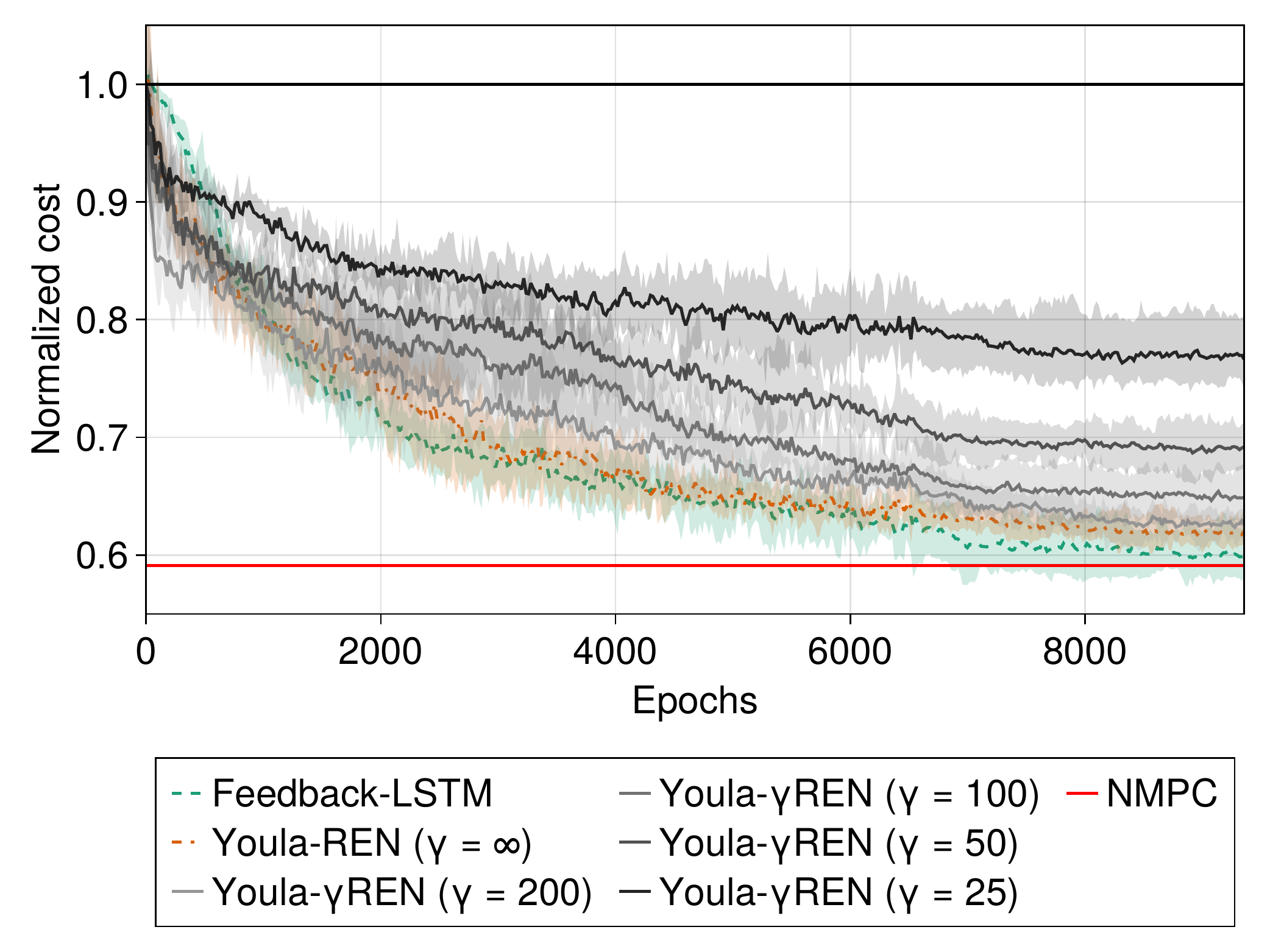}
        \caption{Rotary-arm pendulum.}
        \label{fig:qube-learning}
    \end{subfigure}
    \caption{Loss curves showing the mean test cost vs. training epochs for policies trained on (a) the magnetic suspension problem and (b) the rotary-arm pendulum problem. Colored bands show the range over the six random seeds for each model, lines show the mean. Test cost is normalized by the base controller cost, indicated by the solid black line at 1.0.}
    \label{fig:cost-curves}
\end{figure*}

\subsubsection{Learning performance}

Fig.~\ref{fig:cost-curves} shows the mean test cost for each policy trained on the magnetic suspension and rotary-arm pendulum RL tasks. Policies were benchmarked against nonlinear model predictive controllers (NMPC), which are easy to design for low-dimensional problems. All of the learned policies show significant performance improvements over the base controller, and the best performing models reach the NMPC benchmarks. In particular, we see comparable performance between the Youla-REN and the Feedback-LSTM policy classes, with the Youla-REN achieving a lower cost on magnetic suspension, and the Feedback-LSTM performing slightly better on the rotary-arm pendulum. Together with the results of Sec.~\ref{sec:theory}, we therefore have a policy class that can perform as well as existing state-of-the-art methods on RL tasks for partially-observed nonlinear systems, while also providing stability guarantees for every policy trialled during training. Note that we have not compared the Youla-REN to an LSTM or REN in direct feedback without a base controller (the typical RL policy architecture) in Fig.~\ref{fig:cost-curves}. Training models in this configuration took significantly more epochs than the Youla and Feedback architectures, and achieved a worse final cost than even the base controllers on both tasks.

\subsubsection{Robustness}

One of the great advantages of the Youla-REN policy class is that we can control the performance-robustness trade-off by imposing a Lipschitz upper bound on the REN. The left panels in Figs.~\ref{fig:mag-robust} and \ref{fig:qube-robust} show the effect of perturbing each trained policy with additive adversarial attacks of increasing size on the measurement signal $y_t$. The scatter plots to the right show the attack size required to meaningfully perturb each closed-loop system. We define a ``critical'' attack as one that shifts 1) the average ball position more than 1\,cm from the target and 2) the average pendulum angle more than $30^\circ$ from the vertical in the magnetic suspension and rotary-arm pendulum tasks, respectively. Adversarial attacks were computed with minibatch gradient descent over a receding horizon of 10 time samples. 

Comparing Figs.~\ref{fig:cost-curves} and \ref{fig:robustness} demonstrates the effect of the Lipschitz bound on the performance-robustness trade-off. In Fig.~\ref{fig:cost-curves}, imposing a stronger Lipschitz bound (smaller $\gamma$) drives the Youla-{\gren} policies to worse final costs. In Fig.~\ref{fig:robustness}, however, stronger Lipschitz bound can reliably lead to policies which are more robust to adversarial attacks, even if they perform worse in the unperturbed case. In particular, Fig.~\ref{fig:qube-robust} shows that the base controller and Feedback-LSTM policies are highly sensitive to adversarial attacks in the rotary-arm pendulum environment. We suspect this is because the system has enough degrees of freedom to exhibit chaotic motion, and can be driven to extremely unstable closed-loop responses that are more difficult to recover from than in the magnetic suspension environment. Note that the relationship is not linear --- imposing too strong a Lipschitz bound can lead to less robust policies (for example $\gamma = 25$ in Fig.~\ref{fig:qube-robust}). In practice, careful tuning of the imposed Lipschitz upper bound is required for a given problem.

\subsubsection{Key results}

These results emphasize the strength of the Youla-REN policy class in learning-based control tasks. We can take an existing stabilizing controller for a (nonlinear) dynamical system and learn a robust stabilizing feedback controller that improves some user-defined performance metric. Moreover, we can search over a space of contracting and Lipschitz closed-loops, guaranteeing stability at every step of the training process, and still achieve similar performance to existing methods which provide no such guarantees. Our approach is intuitive in that we can balance the performance-robustness trade-off by tuning the Lipschitz bound of the REN. It is versatile since we do not require special solution methods or projections during training. We therefore expect the Youla-REN to be well-suited to learning-based control in safety-critical robotic systems where performance and robustness are crucial to successful operation.

% Robustness
\begin{figure*}[!t]
    \centering
    \begin{subfigure}[b]{0.69\textwidth}
        \centering
        \includegraphics[width=\textwidth]{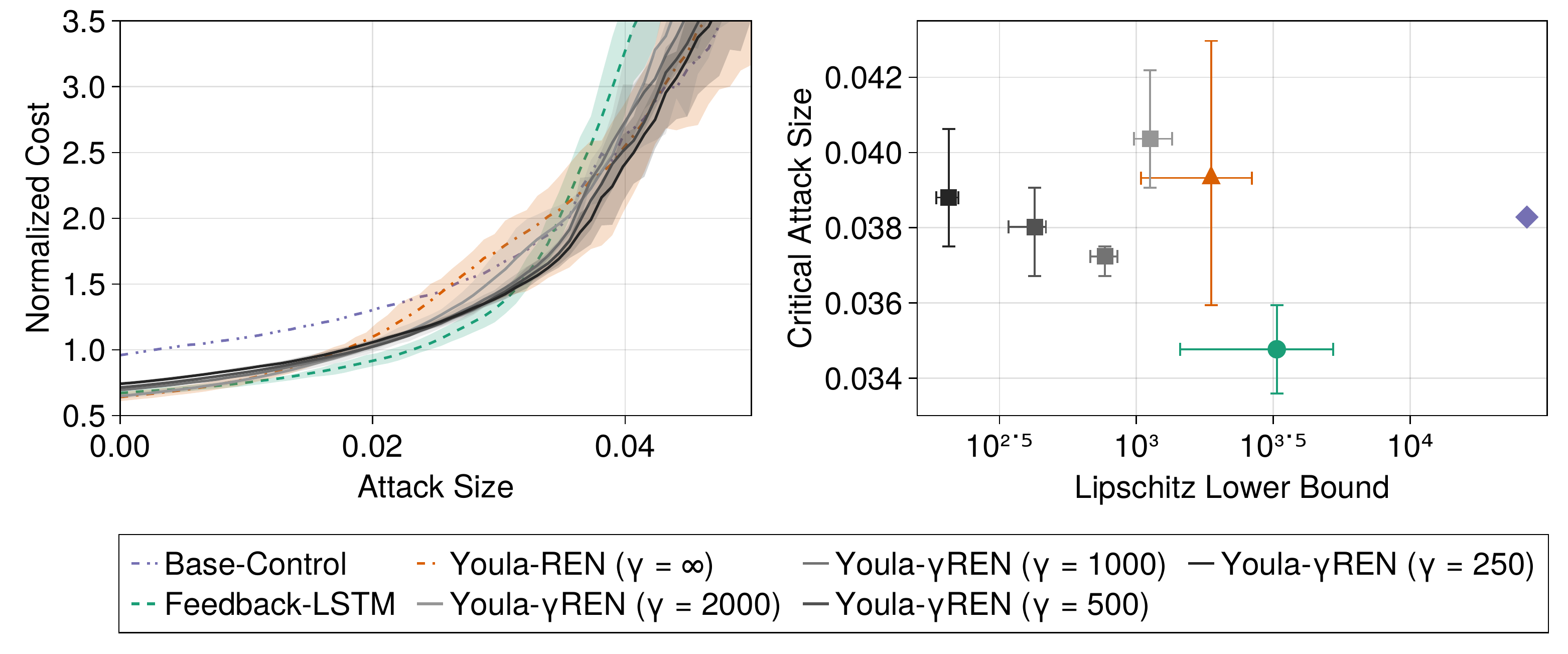}
        \caption{Magnetic suspension.}
        \label{fig:mag-robust}
    \end{subfigure}
    \begin{subfigure}[b]{0.69\textwidth}
        \centering
        \includegraphics[width=\textwidth]{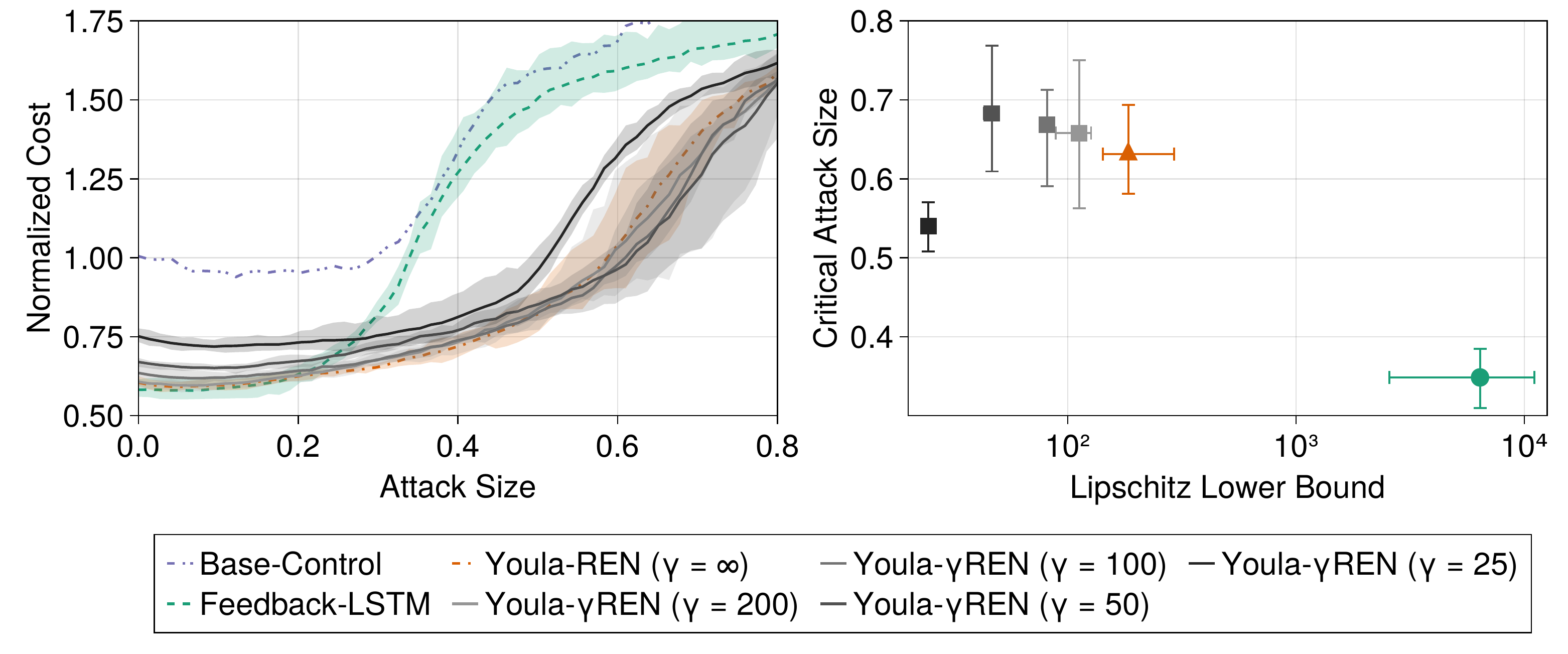}
        \caption{Rotary-arm pendulum.}
        \label{fig:qube-robust}
    \end{subfigure}
    \caption{Normalized test cost vs. adversarial attack size for each of the trained models in Fig.~\ref{fig:cost-curves}. Lines show the average cost over the six random seeds for each model, with bands showing the range. Test cost is normalized by the cost of the (unperturbed) base controller. Scatter plots show the ``critical'' attack size required to perturb (a) the mean ball position 1\,cm from its target and (b) the mean pendulum angle to $30^\circ$ from the vertical in each task, respectively. Error bars are for the six different random seeds for each model. The estimated Lipschitz lower bound of the base controller in (b) is extremely large and has been omitted from the scatter plot. Values of $\gamma$ in the legend are Lipschitz upper bounds.}
    \label{fig:robustness}
\end{figure*}

% Acknowledgements
\section*{Acknowledgements}
The authors would like to thank Professor Alexandre Megretski for his insightful discussions in revising this paper.

% --------------------------------------------------
%
% Appendix
%
% --------------------------------------------------

\appendix
\section*{Proofs of theorems} 
We first prove three auxiliary results. Each of these results are well-known and commonly referred to in the literature in various forms \cite{Lohmiller+Slotine1998,Tsukamoto++2021,Bullo2022,Tran++2019}. We provide proofs of the particular discrete-time statements required in the proof of Theorem~\ref{thm:forward} for completeness.

% Uniform (incremental) exponential convergence
\begin{lemma} \label{lem:contracting-uniform}
Consider a contracting system $\mathcal{T}$ given by \eqref{eqn:generic-system}. Then for any fixed input sequence $u \in \ell^m_{2e}$, any two state trajectories $x^1, x^2 \in \ell_{2e}^n$ exponentially converge to each other with a uniform rate and overshoot.
\end{lemma}
\begin{proof}
Since $\mathcal{T}$ is contracting, it satisfies \eqref{eqn:contract-1} and \eqref{eqn:contract-2} with $c_2 \ge c_1 > 0$ and $\alpha \in [0,1)$.
Repeatedly applying \eqref{eqn:contract-2} gives $V(x^1_t, x^2_t) \le \alpha^t V(x^1_0, x^2_0)$ for all $t\in\mathbb{N}$ where $x^1_0, x^2_0 \in \mathbb{R}^n$ are the initial states of $x^1, x^2$. Applying the upper and lower bounds in \eqref{eqn:contract-1} then dividing through by $c_1$, we have
\begin{equation}
    \abs{x^1_t - x^2_t} \le \beta \alpha^\frac{t}{2} \abs{x^1_0 - x^2_0} \ \forall t\in\mathbb{N}
\end{equation}
where $\beta = \sqrt{c_2/c_1}$ is the overshoot.
\end{proof}

% State trajectories are bounded
\begin{lemma} \label{lem:bounded-states}
    Consider a contracting system $\mathcal{T}$ given by \eqref{eqn:generic-system}. Then for any $x_0 \in \mathbb{R}^n$ and bounded inputs $u\in\ell_{2e}^m$, the states $x_t$ will also be bounded for all $t\in\mathbb{N}$.
\end{lemma}
\begin{proof}
Re-write the dynamics of \eqref{eqn:generic-system} as 
\begin{equation} \label{eqn:dynamics-plusd}
    x_{t+1} = \tilde{f}(x_t) + d(x_t,t)
\end{equation}
where $\tilde{f}(x_t) := f(x_t, 0)$ and $d(x_t, t) := f(x_t, u_t) - f(x_t, 0)$. We note two key facts about the system \eqref{eqn:dynamics-plusd}:
\begin{enumerate}
    \item The system $x_{t+1} = \tilde{f}(x_t)$ is an autonomous (time-invariant) contracting system and thus it will converge to a unique and bounded equilibrium $x^\star \in \mathbb{R}^n$ \cite[Prop.~(v), Sec.~3.7]{Lohmiller+Slotine1998}.
    \item Since $u$ is bounded and $f$ is Lipschitz by assumption \ref{A4}, then $d$ is also bounded and there exists $\bar{d}\in\mathbb{R}^n$ such that $\abs{d(x_t,t)} \le \bar{d}$ for all $t\in\mathbb{N}$.
\end{enumerate}

Let $x^1, x^2\in\ell_{2e}^n$ be solutions of the unperturbed dynamics $x_{t+1} = \tilde{f}(x_t)$ and \eqref{eqn:dynamics-plusd} (respectively) with initial states $x^1_0, x^2_0 \in \mathbb{R}^n$. We know that $x^1_t \rightarrow x^*$ uniformly and exponentially by Lemma~\ref{lem:contracting-uniform}. Therefore, it remains to prove that $\abs{x^1_t - x^2_t}$ is bounded for all $t\in\mathbb{N}$ if $u$ is bounded.

We apply the result of \cite[Thm.~2.8]{Tsukamoto++2021} on \eqref{eqn:dynamics-plusd} --- that bounded additive perturbations to a contracting system generate bounded perturbations to its states --- simplified here for an autonomous contracting system. Suppose there exists a uniformly positive-definite matrix $M(x_t) = \Theta(x_t)^\top \Theta(x_t) \succ 0 \ \forall \, x_t \in \mathbb{R}^n$, where $\Theta: \mathbb{R}^n \rightarrow \mathbb{R}^{n\times n}$ is a nonsingular matrix-valued function, that satisfies
\begin{align}
    &a_1 I \preceq M(x_t) \preceq a_2 I \label{eqn:Tsukamoto-1}\\
    &{\pdv{\tilde{f}}{x}}(x_t)^\top M(x_{t+1}) {\pdv{\tilde{f}}{x}}(x_t) \preceq \sigma^2 M(x_t) \label{eqn:Tsukamoto-2}
\end{align}
where $a_2 \ge a_1 > 0$ and $\sigma \in [0,1)$. Then by \cite[Thm.~2.8]{Tsukamoto++2021},
\begin{equation} \label{eqn:Tsukamoto-bounded}
    \abs{x^1_t - x^2_t} \le C_1 \sigma^t + \frac{\bar{d}}{1 - \sigma} C_2
\end{equation}
for some $C_1, C_2 \in \mathbb{R}^+$ where $C_1$ depends on the initial conditions $x^1_0, x^2_0$.

A function $M(x_t)$ satisfying \eqref{eqn:Tsukamoto-1} and \eqref{eqn:Tsukamoto-2} exists and is well-defined for the autonomous contracting system $x_{t+1} = \tilde{f}(x_t)$, as explained below. Since the system is contracting, there exists a smooth function $V$ satisfying \eqref{eqn:contract-1} and \eqref{eqn:contract-2} with $c_2 \ge c_1 > 0$ and $\alpha \in [0,1)$. We can re-write \eqref{eqn:contract-2} as 
\begin{equation*}
    V(x^1_{t+1}, x^2_{t+1}) - V(x^1_t, x^2_t) \le -c_3 \abs{x^1_t - x^2_t}^2
\end{equation*}
using the lower bound in \eqref{eqn:contract-1}, where $c_3 = (1 - \alpha)c_1$. Hence by \cite[Thm.~11]{Tran++2019} and \cite[Thm.~23.3]{Rugh1996} (both simplified for an autonomous system) there exists a nonsingular matrix-valued function $\Theta: \mathbb{R}^n \rightarrow \mathbb{R}^{n \times n}$ and constants $\mu, \eta, \rho \in \mathbb{R}^+$ such that
\begin{align}
    \eta I \preceq \Theta(x_t)^\top \Theta(x_t) \preceq \rho I \label{eqn:Tran-1}\\
    F(x_t)^\top F(x_t) - I \preceq -\mu I \label{eqn:Tran-2}
\end{align}
for all $x_t\in\mathbb{R}^n$, where $\mu \in [0,1)$ and $F(x_t)$ is given by
\begin{equation}
    F(x_t) = \Theta(x_{t+1}) {\pdv{\tilde{f}}{x}} (x_t) \Theta(x_t)^{-1}.
\end{equation}
Defining $M(x_t) = \Theta(x_t)^\top \Theta(x_t)$, it is clear that \eqref{eqn:Tsukamoto-1} and \eqref{eqn:Tran-1} are identical if $a_1 = \eta, a_2 = \rho$, and \eqref{eqn:Tsukamoto-2} and \eqref{eqn:Tran-2} are equivalent if $\sigma^2 = 1 - \mu$. Therefore, \eqref{eqn:Tsukamoto-bounded} holds and the states of a contracting system $\mathcal{T}$ given by \eqref{eqn:generic-system} are bounded for bounded inputs $u\in\ell_{2e}^m$.
\end{proof}

% States of a contracting system converge exponentially given exponentially converging inputs
\begin{lemma} \label{lem:exp-convergence}
Consider a contracting system 
\begin{equation} \label{eqn:contracting-system}
    \begin{cases}
        x_{t+1} = f(x_t, u_t, r_t) \\
        y_t = h(x_t, u_t, r_t)
    \end{cases}
\end{equation}
with state $x_t \in \mathbb{R}^n$, inputs $u_t\in\mathbb{R}^{m_1}$, $r_t \in \mathbb{R}^{m_2}$, and outputs $y\in\mathbb{R}^p$, where $f$ and $h$ are Lipschitz continuous. Further consider any fixed, bounded $r \in \ell_{2e}^{m_2}$ and any exponentially converging, bounded $u, v \in \ell_{2e}^{m_1}$ such that $\abs{u_t - v_t} \le b a^t$ with $b\in\mathbb{R}^+$ and $a\in[0,1)$. Then the corresponding state and output trajectories $x,z\in\ell_{2e}^n$ and $y,p\in\ell_{2e}^p$ exponentially converge to each other (respectively) and satisfy
\begin{align}
    \abs{x_t - z_t} &\le \beta_1(x_0,z_0) \sigma_1^t \label{eqn:exp-conv1} \\
    \abs{y_t - p_t} &\le \beta_2(x_0,z_0) \sigma_2^t \label{eqn:exp-conv2}
\end{align}
where $\beta_i : \mathbb{R}^n \times \mathbb{R}^n \rightarrow \mathbb{R}^+$ and $\sigma_i \in [0,1)$ for $i \in\{1,2\}$.
\end{lemma}
\begin{proof}
Since the system is contracting, there exists a smooth function $V$ satisfying \eqref{eqn:contract-1} and \eqref{eqn:contract-2} with $c_2 \ge c_1 > 0$ and $\alpha \in [0,1)$. First note that
\begin{align*}
V(x_{t+1}, z_{t+1}) &= V(f(x_t, u_t, r_t), f(z_t, v_t, r_t)) \\
&\le V(f(x_t, u_t, r_t), f(z_t, u_t, r_t)) + \abs{\Delta V}
\end{align*}
where we introduced $\Delta V = V(f(x_t, u_t, r_t), f(z_t, v_t, r_t)) - V(f(x_t, u_t, r_t), f(z_t, u_t, r_t))$. 
$V$ is locally Lipschitz (it is continuously differentiable) and all states of \eqref{eqn:contracting-system} are bounded as per Lemma~\ref{lem:bounded-states}. Hence, there exists some finite upper bound $\gamma_V\in \mathbb{R}^+$ for a given $z, u, v, r$ such that 
$$\abs{\Delta V} \le \gamma_V \abs{f(z_t, v_t, r_t) - f(z_t, u_t, r_t)} \quad \forall\,t\in\mathbb{N}.$$ 
Since $f$ is globally Lipschitz then $\abs{\Delta V} \le \gamma_V\gamma_f \abs{v_t - u_t}$ for some $\gamma_f \in \mathbb{R}^+$
and so $\abs{\Delta V} \le \bar{b} a^t$ where $\bar{b} = b \gamma_V\gamma_f$. Applying \eqref{eqn:contract-2}, we therefore have
\begin{equation} \label{eqn:v-convergence}
    V(x_{t+1}, z_{t+1}) \le \alpha V(x_t, z_t) + \bar{b} a^t.
\end{equation}
It is clear from \eqref{eqn:v-convergence} that $V$ is upper-bounded by the particular solution of a two-state linear system 
$$q_{t+1} = \mqty[\alpha & \bar{b}\\ 0 & a] q_t \quad \text{with} \quad q_0 = \mqty[V(x_0,z_0)\\ 1].$$
This system is stable since $\alpha, a \in [0,1)$, and hence $V$ will be exponentially upper-bounded for all $t$. In fact, we can compute the analytic solution of this system. Using the fact that $\sum_{k=0}^{t-1} \alpha^{t-k} a^k = (\alpha^t - a^t)/(\alpha - a)$, we have
\begin{equation} \label{eqn:v-convergence2}
    V(x_t, z_t) \le \alpha^t V(x_0, z_0) + \bar{b}\cdot \frac{\alpha^t - a^t}{\alpha - a}.
\end{equation}
Applying the upper and lower-bounds on $V$ and dividing through by $c_1$ shows that
\begin{equation} \label{eqn:states-converge}
    \abs{x_t - z_t}^2 \le \frac{c_2}{c_1} \alpha^t \abs{x_0 - z_0}^2 + \frac{\bar{b}}{c_1}\cdot \frac{\alpha^t - a^t}{\alpha - a}.
\end{equation}
Taking the square root gives an inequality in the form of \eqref{eqn:exp-conv1}, where $\beta_1 : \mathbb{R}^n \times \mathbb{R}^n \rightarrow \mathbb{R}^+$ and $\sigma_1 \in [0,1)$ take different values depending on $\alpha$ and $a$:
\begin{enumerate}
    \item If $\alpha \ne a$, then \eqref{eqn:exp-conv1} is satisfied with $\sigma_1 = \sqrt{\max(\alpha, a)}$ and $\beta_1^2 = \frac{c_2}{c_1}\abs{x_0 - z_0}^2 + \frac{\bar{b}}{c_1\abs{\alpha - a}}.$
    \item If $\alpha = a$, the sum of terms giving \eqref{eqn:v-convergence2} collapses to $\sum_{k=0}^{t-1} \alpha^t = t\alpha^t$ and \eqref{eqn:states-converge} becomes
    $$\abs{x_t - z_t}^2 \le \left(\frac{c_2}{c_1} \abs{x_0 - z_0}^2 + \frac{\bar{b}}{c_1} t \right) \alpha^t .$$
    The function on the right-hand side can be uniformly and exponentially upper-bounded by $\beta_1 \sigma_1^t$ where $\sigma_1 \in [\alpha, 1)$, $\beta_1 \in \mathbb{R}^+$ since the exponential function $(\sigma_1/\alpha)^t$ dominates the linear term for sufficiently large $t$. Hence there exists a $T \in \mathbb{R}^+$ such that \eqref{eqn:exp-conv1} is satisfied with $\beta_1^2 = \frac{c_2}{c_1} \abs{x_0 - z_0}^2 + \frac{\bar{b}}{c_1} T$ and $\sigma_1 \in [\alpha, 1)$.
\end{enumerate}

The result \eqref{eqn:exp-conv2} follows from \eqref{eqn:exp-conv1} and $\abs{u_t - v_t} \le b a^t$ by noting that the output function $h$ in \eqref{eqn:contracting-system} is Lipschitz.
\end{proof}

% Proof of Theorem 1
\noindent \begin{proof}\textit{(Theorem~\ref{thm:forward})}
Let $\bar{x}_t = (x_t^\top, \tilde{x}_t^\top, q_t^\top)^\top$ be the state of the closed-loop system under $\mathcal{K}_\mathcal{Q}$, where $\tilde{x}_t = x_t - \hat{x}_t$ is the observer error. The closed-loop system maps $r \mapsto z$ with
\begin{equation} \label{eqn:cl-dynamics}
\mathcal{G}_\mathrm{CL}:
\begin{cases}
    x_{t+1} = f(x_t, \bar{u}_t) \\
    \tilde{x}_{t+1} = f(x_t, \bar{u}_t) - f_o(x_t - \tilde{x}_t, \bar{u}_t, c(x_t)) \\
    q_{t+1} = f_q(q_t, r_t, c(x_t) - c(x_t - \tilde{x}_t)) \\
    z_t = (x_t^\top, (k(x_t - \tilde{x}_t) + \tilde{u}_t)^\top)^\top
\end{cases}
\end{equation}
where $\bar{u}_t = k(x_t - \tilde{x}_t) + \tilde{u}_t + r_t$ and $\tilde{u}_t = h_q(q_t, r_t, c(x_t) - c(x_t - \tilde{x}_t))$. We will show that $\mathcal{G}_\mathrm{CL}$ is contracting and Lipschitz with transients due to initial conditions as per Definitions~\ref{dfn:contraction-transients} and \ref{dfn:lipschitz-sys}, respectively.

% Contraction
\paragraph{Proof of contraction}
Consider the observer error $\tilde{x}$ for a single trajectory. Since the observer satisfies the correctness assumption \ref{A2}, we can write $\tilde{x} = \mathcal{O}_{x_0}(\bar{u}, y) - \mathcal{O}_{\hat{x}_0}(\bar{u}, y)$, where the subscripts on $\mathcal{O}$ distinguish between the initial system and observer states. The inputs $\bar{u}=u + r$ and $y$ are the same for both the true and estimated state trajectories $x$ and $\hat{x}$ (respectively). Therefore by Lemma~\ref{lem:contracting-uniform}, there exists a $\beta \in \mathbb{R}^+$ and $\alpha \in [0,1)$ such that $\abs{\tilde{x}_t} \le \beta \alpha^t \abs{\tilde{x}_0}$ since the observer is contracting by \ref{A3}.

Now consider any two state trajectories $\bar{x}^1, \bar{x}^2 \in \ell_{2e}$ of the closed-loop system $\mathcal{G}_\mathrm{CL}$ starting from initial states $\bar{x}^1_0, \bar{x}^2_0 \in \mathbb{R}^{2n+q}$ and given the same input sequence $r\in\ell^m_{2e}$. The observer errors exponentially converge to one another since
\begin{align*}
    \abs{\tilde{x}^1_t - \tilde{x}^2_t} \le \beta \alpha^t (\abs{\tilde{x}^1_0} + \abs{\tilde{x}^2_0}) =: \bar{\beta}(\tilde{x}^1_0,\tilde{x}^2_0) \alpha^t.
\end{align*}
The same is true for the innovations signals
\begin{equation}
    \abs{\tilde{y}^1_t - \tilde{y}^2_t} \le \gamma_c \bar{\beta}(\tilde{x}^1_0,\tilde{x}^2_0) \alpha^t,
\end{equation}
where $\gamma_c$ is the Lipschitz bound of the measurement function $c(\cdot)$ from \eqref{eqn:system-G}.

We now repeatedly apply Lemma~\ref{lem:exp-convergence} to prove the result. Since $\tilde{u} = \mathcal{Q}(r,\tilde{y})$ is contracting, $\abs{\tilde{y}^1_t - \tilde{y}^2_t}$ exponentially converges to zero, and $r^1 = r^2$ is bounded, then by Lemma~\ref{lem:exp-convergence}
\begin{equation*}
    \abs{\tilde{u}^1_t - \tilde{u}^2_t} \le \beta_q(q^1_0, q^2_0, \tilde{x}^1_0, \tilde{x}^2_0) \alpha_q^t
\end{equation*}
with $\beta_q \ge 0$ and $\alpha \in [0,1)$. By assumption \ref{A1}, the closed-loop system under the base controller mapping $\tilde{u},r \mapsto z$ is also contracting. Therefore, the states $(x_t^{i\top}, \tilde{x}_t^{i\top})^\top$ with $i\in\{1,2\}$ also exponentially converge to one another with uniform rate and an overshoot dependent on initial conditions, again by Lemma~\ref{lem:exp-convergence}. Hence there exists a $\beta_\mathrm{CL}(\bar{x}^1_0, \bar{x}^2_0) \ge 0$ and $\alpha_\mathrm{CL} \in [0,1)$ such that 
\begin{equation}
    \abs{\bar{x}^1_0 - \bar{x}^2_0} \le \beta_\mathrm{CL}(\bar{x}^1_0, \bar{x}^2_0) \alpha_\mathrm{CL}^t,
\end{equation}
so $\mathcal{G}_\mathrm{CL}$ is contracting with transients.

% Lipschitz
\paragraph{Proof of Lipschitz}
Consider two trajectories $z^1, z^2 \in \ell_{2e}$ of the closed-loop system starting from initial states $\bar{x}^1_0, \bar{x}^2_0 \in \mathbb{R}^{2n+q}$ with inputs $r_1, r_2 \in \ell^m_{2e}$, respectively. Denote $\Delta z = z^1 -z^2$ and similarly for all other variables.
By assumption \ref{A1}, the closed-loop system under the base controller is Lipschitz with respect to any inputs $(r+\tilde{u})\in\ell_{2e}^m$, hence $\exists \ \gamma \in \mathbb{R}^+$ and $\kappa_1(s^1_0, s^2_0) \ge 0$ such that
\begin{align*}
    \norm{\Delta z}_T &\le \gamma \norm{(r^1 + \tilde{u}^1) - (r^2 + \tilde{u}^2)}_T + \kappa_1(s^1_0, s^2_0) \\
    &\le \gamma \norm{\Delta r}_T + \gamma \norm{\Delta \tilde{u}}_T + \kappa_1(s^1_0, s^2_0).
\end{align*}
Since $\mathcal{Q}$ is Lipschitz by assumption, then similarly $\exists \ \gamma_{q_r}, \gamma_{q_y} \in \mathbb{R}^+$ and $\kappa_2(q^1_0, q^2_0) \ge 0$ such that
$$
\norm{\tilde{u}}_T \le \gamma_{q_r} \norm{\Delta r}_T + \gamma_{q_y} \norm{\Delta \tilde{y}}_T + \kappa_2(q^1_0, q^2_0)
$$
Note further that $\tilde{y} = c(x) - c(\hat{x})$ (where $c(\cdot)$ acts element-wise on signals) so $\norm{\Delta \tilde{y}}_T \le \gamma_c \norm{\Delta \tilde{x}}_T$ since $c(\cdot)$ is Lipschitz with bound $\gamma_c \in \mathbb{R}^+$ by assumption \ref{A4}. Combining these expressions, we have that
\begin{equation} \label{eqn:lip-z-halfway}
    \norm{\Delta z}_T \le \gamma (1 + \gamma_{q_r}) \norm{\Delta r}_T + \tilde{\gamma} \norm{\Delta \tilde{x}}_T + \kappa_3(\bar{x}^1_0, \bar{x}^2_0)
\end{equation}
where $\tilde{\gamma} = \gamma \gamma_{q_y} \gamma_c$ and $\kappa_3 = \kappa_1 + \gamma \kappa_2$. Since the observer error satisfies $\abs{\tilde{x}_t} \le \beta \alpha^t \abs{\tilde{x}_0}$ for any $\tilde{u},r$, then there exists a finite $\Gamma(\tilde{x}_0) \ge 0$ such that $\norm{\tilde{x}}_T \le \Gamma(\tilde{x}_0)$. Hence $\norm{\Delta \tilde{x}}_T \le \Gamma(\tilde{x}^1_0) + \Gamma(\tilde{x}^2_0)$ and substituting into \eqref{eqn:lip-z-halfway} gives
\begin{equation} \label{eqn:lip-proof}
    \norm{\Delta z}_T \le \bar{\gamma} \norm{\Delta r}_T + \kappa_4(\bar{x}^1_0, \bar{x}^2_0)
\end{equation}
where $\bar{\gamma} = \gamma (1 + \gamma_{q_r})$ and $\kappa_4 = \tilde{\gamma}(\Gamma(x^1_0, \hat{x}^1_0) + \Gamma(x^2_0, \hat{x}^2_0)) + \kappa_3$. The closed-loop system $\mathcal{G}_\mathrm{CL}$ is therefore Lipschitz with the effect of initial conditions captured by $\kappa_4 \ge 0$.
\end{proof} % End proof of theorem 1

% Proof of theorem 2
\begin{proof} \textit{(Theorem~\ref{thm:reverse})}
Consider the closed-loop system of the disturbed plant $\mathcal{G}_d$ from \eqref{eqn:plant-disturbed} in feedback with a robustly stabilizing feedback controller
    \begin{equation} \label{eqn:K-stable}
        \mathcal{K}: 
        \begin{cases}
            \phi_{t+1} = h(\phi_t, y_t) \\
            u_t = g(\phi_t, y_t)
        \end{cases}
    \end{equation}
with states $\phi_t\in\mathbb{R}^k$. The closed-loop system $\bar{\mathcal{G}}_\mathrm{CL}: (r, d_x, d_y) \mapsto z$ is given by
\begin{equation} \label{eqn:cl-perturbed}
\bar{\mathcal{G}}_\mathrm{CL}:
\begin{cases}
    x_{t+1} = f(x_t, g(\phi_t, c(x_t) + d_{y_t}) + r_t) + d_{x_t} \\
    \phi_{t+1} = h(\phi_t, c(x_t) + d_{y_t}) \\
    z_t = (x_t^\top, g(\phi_t, c(x_t) + d_{y_t})^\top)^\top.
\end{cases}
\end{equation}
The system $\bar{\mathcal{G}}_\mathrm{CL}$ is contracting and Lipschitz by assumption. We now aim to show that $\mathcal{K}$ can be represented by a Youla controller \eqref{eqn:youla-ctrl} -- \eqref{eqn:youla-param} parameterized by a contracting and Lipschitz $\mathcal{Q}_\mathcal{K}$. By augmenting $\mathcal{K}$ with a Lipschitz state-feedback controller $k(\cdot)$ and an observer satisfying assumptions $\ref{A2}$ to \ref{A4}, we can construct a map $\mathcal{Q}_\mathcal{K}: (r, \tilde{y}) \mapsto \tilde{u}$ using the fact that $y_t = \tilde{y}_t + c(\hat{x}_t)$ and $\tilde{u}_t = \bar{u}_t - k(\hat{x}_t)$. This gives
\begin{equation}
    \mathcal{Q}_\mathcal{K}:
    \begin{cases}
        \hat{x}_{t+1} = f_o(\hat{x}_t, g(\phi_t, c(\hat{x}_t) + \tilde{y}_t) + r_t, c(\hat{x}_t)) + \Delta f_o \\
        \phi_{t+1} = h(\phi_t, c(\hat{x}_t) + \tilde{y}_t) \\
        \tilde{u}_t = g(\phi_t, c(\hat{x}_t) + \tilde{y}_t) - k(\hat{x}_t).
    \end{cases}
\end{equation}
where $\Delta f_o := f_o(\hat{x}_t, u_t + r_t, c(\hat{x}_t) + \tilde{y}) - f_o(\hat{x}_t, u_t + r_t, c(\hat{x}_t))$.
Note that $\mathcal{Q}_\mathcal{K}$ does not change the control signal $u_t$ since $\tilde{u}_t + k(\hat{x}_t) = g(\phi_t, y_t)$ by construction. We show that $\mathcal{Q}_\mathcal{K}$ is contracting and Lipschitz by comparing it to $\bar{\mathcal{G}}_\mathrm{CL}$. Make the substitution $\hat{x}_t \leftrightarrow x_t$, $\tilde{y}_t \leftrightarrow d_{y_t}$, $\Delta f_o \leftrightarrow d_{x_t}$. Then the state dynamics of $\mathcal{Q}_\mathcal{K}$ and $\bar{\mathcal{G}}_\mathrm{CL}$ are identical after applying assumption \ref{A2}, and hence $\mathcal{Q}_\mathcal{K}$ is contracting.

It remains to show that $\mathcal{Q}_\mathcal{K}:(r,\tilde{y}) \rightarrow \tilde{u}$ is Lipschitz. We know $(r, d_x, d_y) \mapsto z$ is Lipschitz, and hence so too is $(r, \tilde{y}) \mapsto z$ following our transformation, noting that the norm of $\Delta f_o$ is linearly bounded by $\tilde{y}$ since $f_o$ is Lipschitz. Considering $z \mapsto \tilde{u}$, note that under $x \leftrightarrow \hat{x}$ then $z = (\hat{x}^\top, u^\top)^\top$ and $u = \tilde{u} + k(\hat{x})$ so $\tilde{u} = k([I, 0] z) - [0, I]z$ where $k(\cdot)$ acts element-wise on signals. Since $k(\cdot)$ is a Lipschitz function, then the map $z \mapsto \tilde{u}$ is Lipschitz and hence so too is $\mathcal{Q}_\mathcal{K}: (r, \tilde{y}) \mapsto \tilde{u}$.
\end{proof} % End proof of Theorem 2

% --------------------------------------------------------------------
% Training details
% --------------------------------------------------------------------
\section*{Model configurations and training details}

We compared the three model architectures on each RL task outlined in Section~\ref{sec:numeric-experiments}. We selected four Lipschitz bounds for the Youla-{\gren}s on each task to compare the effect of imposing robustness constraints. Similar numbers of model parameters were used for a fair comparison. LSTM models were allocated 28 cell units, while the RENs and {\gren}s were given 32 states and 64 neurons. We found that reducing the dimensionality of one of the REN parameters (the $X$ matrix in \cite[Eqn.~23]{Revay++2023}) from ($128\times128$) to ($16\times128$) accelerated learning and improved performance in the Youla-REN and Youla-{\gren}s. ReLU activation functions were used for all RENs models.

We chose to use a separate, de-tuned observer to compute the innovations in the magnetic suspension task instead of the base observer. This slowed the contraction rate of the innovations sequence, giving the network more samples with a non-zero input signal to affect the controls. This does not violate any assumptions outlined in Sec.~\ref{sec:theory}, since our only requirements for the observer were only that it satisfied the correctness, contraction, and Lipschitz properties \ref{A2} to \ref{A4}.

We trained our models using a version of the \textit{Augmented Random Search} (ARS)-v1 algorithm from \cite{Mania++2018}. Each model was trained with six random seeds to account for variability in the initialization. We averaged gradient estimates over 16 perturbation directions at each training epoch, using batches of 50 random initial conditions to approximate the expected costs, and clipped gradients to an $\ell^2$ norm of 10. A batch size of 100 was used for the test cost. Models were trained with the ADAM optimizer \cite{Kingma+Ba2015}. Learning rates and ARS exploration magnitudes were tuned by sweeping over a wide range from $10^{-5}$ to $10^{-1}$. The best-performing combinations for each model and task are provided in Table~\ref{tab:learning-params}. Learning rates were decreased by a factor of 10 after 70\% of the total epochs to verify that the models had converged.

\begin{table}[!t]
    \vspace{3mm}
    \centering
    \begin{tabular}{|l|c|c|c|c|c|c|}
    \hline
    & \multicolumn{3}{c|}{Magnetic Suspension} & \multicolumn{3}{c|}{Rotary Pendulum} \\
    \hline
    Model & $\gamma$ & $\alpha$ & $\sigma$ & $\gamma$ & $\alpha$ & $\sigma$ \\
    \hline
    Feedback-LSTM   & -         & 0.01      & 0.01  & -         & 0.01 & 0.02 \\
    Youla-REN       & $\infty$  & 0.005     & 0.01  & $\infty$  & 0.01 & 0.05\\             
    \multirow{4}{*}{Youla-{\gren}}                
                    & 2000      &   0.005   & 0.01  & 200       & 0.01 & 0.05\\
                    & 1000      &   0.005   & 0.05  & 100       & 0.01 & 0.05\\
                    & 500       &   0.005   & 0.05  & 50        & 0.01 & 0.05\\
                    & 250       &   0.005   & 0.05  & 25        & 0.01 & 0.05\\
    \hline
    \end{tabular}
    \caption{Learning rates $\alpha$, ARS exploration magnitudes $\sigma$, and imposed Lipschitz upper bounds $\gamma$ on the two RL tasks.}
    \label{tab:learning-params}
\end{table}

%%%%%%%%%%%%%%%%%%%%%%%%%%%%%%%%%%%%%%%%%%%%%%%%%%%%%%%%%%%%%%%%%%%%%%%%%%%%%%%%

\bibliographystyle{ieeetr}
\bibliography{references}

\end{document}